\documentclass[10pt,reqno]{amsart}

\newcommand\version{August 24, 2022}

\usepackage[ansinew]{inputenc}

\usepackage{textcomp}

\usepackage{cite}

\usepackage{amsmath,amsfonts,amsthm,amssymb,amsxtra}

\usepackage[usenames,dvipsnames]{color}

\usepackage{bbm}
\usepackage{fancyhdr}
\usepackage[scanall]{psfrag}
\usepackage{graphicx}
\usepackage{ifthen}
\usepackage{pstricks,pst-node,pst-plot}

\usepackage{wasysym}

\usepackage[toc,page]{appendix}

\newtheorem{theorem}{Theorem}[section]

\newtheorem{lemma}[theorem]{Lemma}

\theoremstyle{definition}

\newtheorem{example}[theorem]{Example}

\theoremstyle{remark}

\numberwithin{equation}{section}

\renewcommand{\epsilon}{\varepsilon}

\DeclareMathOperator{\dist}{dist}

\DeclareMathOperator{\im}{Im}

\DeclareMathOperator{\Tr}{Tr}

\begin{document}

\title[Perturbation Determinant--- \version]{Perturbation Determinant and levinson's formula for Schr\"odinger operators with generalized Point interaction}

\author{M. Fazeel Anwar, Muhammad Usman and Muhammad Danish Zia}
\address{M. Fazeel Anwar, Department of Mathematics, Sukkur IBA University, Sukkur, Pakistan}
\email{fazeel.anwar@iba-suk.edu.pk}
\address{Muhammad Usman, Department of Mathematics, Lahore University of Management Sciences (LUMS), Lahore, Pakistan}
\email{usman@lums.edu.pk}
\address{Muhammad Danish Zia, Department of Basic Sciences, School of Civil Engineering, National University of Sciences and Technology (NUST), Islamabad, Pakistan}
\email{dazia@mce.nust.edu.pk}


\begin{abstract}
We consider the one dimensional Schr\"odinger operator with properly connecting generalized point interaction at the origin. We derive a trace formula for trace of difference of resolvents of perturbed and unperturbed Schr\"odinger operators in terms of a Wronskian which results into an explicit expression for perturbation determinant. Using the estimate for large time real argument  on the trace norm of the resolvent difference of the perturbed and unperturbed Schr\"odinger operators we express the spectral shift function in terms of perturbation determinant. Under certain integrability condition on the potential function, we calculate low energy asymptotics for the perturbation determinant  and prove an analog of Levinson's formula.
\end{abstract}

\maketitle

\section{Introduction and main results}

Let $H_{V}^{\mathcal{A}}$ be the Schr\"odinger operator 
\begin{equation}\label{theop}
H_{V}^{\mathcal{A}}=H^{\mathcal A}_0+V,\quad H^{\mathcal A}_0=-\frac{d^2}{dx^2}
\end{equation}
on the real-line which is realized as a union of two positive semi-axis $e_j=[0,\infty)$, $j=1,2$ coupled at $0$. The potential function $V$ is assumed to be real-valued and satisfies
\begin{equation}\label{vcond}
\int_{e_j}|V_j(x)|dx<\infty,
\end{equation}
where, $V_j$ is the restriction of $V$ on $e_j$, $j=1,2$.
The operator \eqref{theop} is self-adjoint for all functions from the space $H^2([0,\infty))\oplus H^2([0,\infty))$ satisfying the conditions
\begin{align}\label{gencond}
\begin{bmatrix}
\psi_1(0)\\
\psi_1'(0)
\end{bmatrix}
=
\mathcal{A}\begin{bmatrix}
\psi_2(0)\\
\psi_2'(0)
\end{bmatrix}
\end{align}
where $\mathcal{A}=e^{i\phi}
\begin{bmatrix}
a &b\\
c &d
\end{bmatrix}
$ and $\psi=(\psi_1,\psi_2)^T \in L^2([0,\infty))\oplus L^2([0,\infty))$ and the real parameters $\phi\in [-\pi/2,\pi/2]$, $a,b,c,d\in \mathbb{R}$ satisfying $ad-bc=-1$.\\

The above conditions are commonly known as the generalized point interaction, first introduced by \v Seba \cite{seba}. There are several other equivalent formulations for the generalized point interaction, see for example \cite{exn3} and references therein. The free Hamiltonian $H_0^{\mathcal A}$ with conditions \eqref{gencond} describes certain second order differential operators with generalized functions in coefficients. The Schr\"odinger operator with $\delta$ potential of strength $\alpha$, for instance, corresponds to the unperturbed operator $H_0^{\mathcal A}$ with $\phi=0,\, a=1,\, b=0,\, c=\alpha$ and $d=-1$. Similarly, $\delta'$ potential of strength $\beta$ corresponds to $\phi=0,\, a=-1,\,b=\beta,\,c=0$ and $d=1$.  Differential operators of this kind are closely related to exactly solvable models in quantum mechanics, atomic physics, and acoustics \cite{Alb, ynd}. 
\\

The present article is devoted to deriving explicit expression for perturbation determinant for the pair of operators $H_V^{\mathcal A}$ and $H_0^{\mathcal A}$ and its relationship with spectral shift function and Levinson's theorem. These mathematical objects play an important role in the study of direct and inverse scattering theory \cite{Har2, Har3, Har4, gre} as well as in solid state physics in connection with Friedel sum rule and excess charge \cite{koh}. Kr\u ein \cite{kre} (see also \cite{kuroda}) introduced the concept of perturbation determinant and it is an important tool in studying trace formulas of higher order.   For a detailed study on these topics for Schr\"odinger operators on the half-line as well as on the whole real line we refer to the monograph \cite{yaf} (see also reference \cite{DU}) and for a quantum star graph we refer to \cite{D1}. Our aim in this paper is to generalize these results for Schr\"odinger operators on the real-line, which can also be seen as a two edge star graph, satisfying the most general properly connecting self-adjoint matching conditions at zero. \\

 Using Kr\u ein's resolvent formula (see \cite{Alb} and \cite{exn2}), we first prove a trace formula that expresses the trace of the difference between the perturbed and corresponding unperturbed resolvents in terms of a Wronskian. This trace formula allows us to derive an explicit expression for the perturbation determinant. The perturbation determinant is an interesting object in the spectral theory of Schr\"odinger operators, as the zeros of perturbation determinant coincide with the eigenvalues of the Schr\"odinger operator and the multiplicity of each eigenvalue is equal to the order of the corresponding zero of the perturbation determinant \cite{Ak2}. Estimates on the trace norm of the resolvent difference of the operators $H_V^{\mathcal A}$ and $H_0^{\mathcal A}$  obtained in \cite{rec3} allow one to express the spectral shift function in terms of perturbation determinant. For a real argument, perturbation determinant is used to define the phase shift function. The so called Levinson's formula is derived for the Schr\"odinger operator $H_V^{\mathcal A}$ in terms of its phase shift function. This formula, also called a zero order trace formula gives a relationship between the number of negative eigenvalues and the scattering data via phase shift function. \\

To state our main results we first consider the half-line Schr\"odinger operators $H^D_{e_j}=-\frac{d^2}{dx^2}+V_j$, $j=1,2$,  with Dirichlet boundary condition at $0$ and let $H^D_V=H^D_{e_1}\oplus H^D_{e_2}$ denote the decoupled Schr\"odinger operator on the whole line. The resolvent $(H^D_V-z)^{-1}$ is denoted by $R^D_V(z)$. Under the condition \eqref{vcond} the differential equation
$$
-u_j''+V_ju=\zeta^2 u,\quad j=1,2
$$ 
has two particular solutions, namely, the regular solution and the Jost solution. The regular solution $\phi_j$ is characterized by the conditions
$$
\phi_j(0,\zeta)=0,\quad \phi'_j(0,\zeta)=1
$$
and the Jost solution $\theta_j$ by the asymptotics $\theta_j(x,\zeta)\sim e^{i\zeta x}$, as $x\rightarrow \infty$. We denote the resolvent $(H_{V}^{\mathcal{A}}-z)^{-1}$ of the perturbed operator $H_{V}^{\mathcal{A}}$ by $R^{\mathcal{A}}_V(z)$ and the resolvent $(H^{\mathcal A}_0-z)^{-1}$ of the unperturbed operator $H^{\mathcal A}_0$ by $R^{\mathcal A}_0(z)$. Moreover, the (modified) perturbation determinant  is defined as
$$
D(z):=\mbox{det}(\mathbb{I}+\sqrt{V}R_0^{\mathcal A}(z)\sqrt{|V|}),\quad z\in\rho(H_0^{\mathcal A}),
$$
where, $\sqrt{V}=\mbox{sgn} V\sqrt{|V|}$. 
 \\

Our first main result is the following trace formula for the difference of two resolvents in terms of Jost solutions $\theta_j$ and their derivatives
\begin{theorem}\label{t1}
If the potential $V$ satisfies \eqref{vcond} then the following trace formula holds
\begin{align}\label{traceformula}
\Tr(R^{\mathcal{A}}_{V}(z)-R^{\mathcal A}_{0}(z))=-\dfrac{1}{2\zeta}\left(\dfrac{d}{d\zeta}\ln\left(\dfrac{w_1(\zeta)w_2(\zeta)L(\zeta)}{(a-d)\zeta+(b\zeta^2+c)i}\right)\right), \quad \zeta=z^{1/2},\quad \emph{Im}\,\zeta>0
\end{align}
where $$L(\zeta)=a\frac{\theta_1'(0,\zeta)}{\theta_1(0,\zeta)}-d\frac{\theta_2'(0,\zeta)}{\theta_2(0,\zeta)}+b\frac{\theta_1'(0,\zeta)\theta_2'(0,\zeta)}{\theta_1(0,\zeta)\theta_2(0,\zeta)}-c$$ and $w_j(\zeta)=\theta_j(0,\zeta)$, $j=1,2$. Moreover, for $z\in\rho(H_V)$ and
  $\emph{Im}\,z^{1/2}>0,$ 
  the perturbation determinant $D(\zeta)$ of the operator $H_{V}^{\mathcal{A}}$ with respect to $H^{\mathcal A}_0$ is given by
\begin{equation}\label{PDEx}
D(z)=\dfrac{L(\sqrt{z}) w_1(\sqrt{z})w_2(\sqrt{z})}{(a-d)\sqrt{z}i-(bz+c)}.
\end{equation}
\end{theorem}

In Section 4 we study the behaviour of $D(z)$ as $|z|\rightarrow 0$ and prove the zero order trace formula, commonly known as the Levinson's formula, for the operator $H_{V}^{\mathcal{A}}$. We will need the following constants:
\begin{align}\label{constants}
\alpha_1 & =\left(b\left(\dfrac{\theta_1'(0,0)}{\theta_2(0,0)}+\dfrac{\theta_2'(0,0)}{\theta_1(0,0)}\right)-a\dfrac{\theta_2(0,0)}{\theta_1(0,0)}+d\dfrac{\theta_1(0,0)}{\theta_2(0,0)}\right)i,\\ \nonumber
\alpha_2 &=a\theta_1'(0,0)\theta_2(0,0)+b\theta_1'(0,0)\theta_2'(0,0),\\\nonumber
\alpha_3 &=(c\theta_2(0,0)+d\theta_2'(0,0))\dot{\theta}_1(0,0),\\\nonumber
\alpha_4 &=a\theta_1'(0,0)\dot{\theta}_2(0,0)-d\dot{\theta}_1(0,0)\theta_2'(0,0),
\end{align}
where dot denotes the derivative with respect to $\zeta$.
\begin{theorem}\label{lea}
Assume that

\begin{align}\label{Dec_Condition}
\int\limits_{e_j}(1+x)|V_j(x)|\,dx<\infty,\hspace*{0.3cm}j=1,2
\end{align}
and let $N$ be the number of negative eigenvalues of the operator $H_{V}^{\mathcal{A}}$. Then, the following formula holds
\[ \eta(\infty)-\eta(0)=\begin{cases} 
      \pi \left(N-\dfrac{P-Q}{2}\right), & w_j(0)\neq0$, $j=1,2,\\
      \pi \left(N-\dfrac{P-R}{2}\right), & w_1(0)=0,w_2(0)\neq0,\\
      \pi \left(N-\dfrac{P-S}{2}\right), & w_1(0)\neq0,w_2(0)=0,\\
      \pi \left(N-\dfrac{P-T}{2}\right), & w_j(0)=0$, $j=1,2.
   \end{cases}
\]
Here, $\eta$ is the so-called phase shift function, $P,Q,R,S$ and $T$ are functions which are defined as
\begin{align*}
P &=\begin{cases} 
      0 & c\neq0 \\
      1 & c=0,a-d\neq0\\
      2 & c=0,a-d=0,b\neq0 ,
   \end{cases}\qquad
 Q =\begin{cases} 
      0 & L(0)\neq0 \\
      1 & L(0)=0,\alpha_1\neq0\\
      2 & L(0)=0,\alpha_1=0,b\neq0 ,
   \end{cases}\\
 R &=\begin{cases} 
      0 & \alpha_2\neq0 \\
      1 & \alpha_2=0,
   \end{cases}\qquad
 S =\begin{cases} 
      0 & \alpha_3\neq0 \\
      1 & \alpha_3=0,
   \end{cases}\\
 T&=\begin{cases} 
      0 & b\neq0 \\
      1 & b=0,\alpha_4\neq0\\
      2 & b=0,\alpha_4=0,c\neq0 .
   \end{cases}
\end{align*}
\end{theorem}

\section{Perturbation determinant}
Our main purpose in this section is to derive the trace formula \eqref{traceformula} and the explicit expression \eqref{PDEx} for the perturbation determinant of operator $H_{V}^{\mathcal{A}}$ with respect to $H_0$ in terms of the Jost solutions $\theta_j$ and their derivatives $\theta_j'$. For this purpose we use Kr\u ein's resolvent formula which allows us to describe the kernel of $R_V^{\mathcal A}(z) $ in terms of the kernel of $R^D_V(z)$. More precisely, for $\mbox{Im}\,\zeta\geq 0$, the Krein's formula states
\begin{align}\label{krein}
R^{\mathcal A}_{j,l}(x,y;z):=R_{j,l}^D(x,y;z)+\lambda_{j,l}\theta_j(x,\zeta)\theta_l(y,\zeta),\quad j,l=1,2,\quad z=\zeta^2\in\rho(H^{\mathcal A}_V)\cap \rho(H^D_V),
\end{align}
where $R^{\mathcal A}_{j,l}$ and $R^D_{j,l}$ denote the kernel of $2\times 2$ matrix integral operators $R_V^{\mathcal A}$ and $R^D_V$, respectively. The values of the coefficients $\lambda_{j,l}$ are given by the following lemma
\begin{lemma}\label{lambdavalues}
Let $\psi(x):=\int_0^\infty R^{\mathcal{A}}(x,y;z)f(y)\,dy$ and $f=(f_1(x),f_2(x))^T\in L^2([0,\infty))\oplus L^2([0,\infty))$. If $\psi(x)$ satisfies $H^{\mathcal A}_V\psi=\zeta^2\psi+f$ and conditions \eqref{gencond} then $\lambda_{j,l}$ in \eqref{krein} are given by
\begin{align*}
\begin{bmatrix}
\lambda_{1,1} &\lambda_{1,2}\\
\lambda_{2,1} &\lambda_{2,2}\\
\end{bmatrix}
=
\begin{bmatrix}
-\dfrac{a+b\dfrac{\theta_2'(0,\zeta)}{\theta_2(0,\zeta)}}{L(\zeta)\theta_1^2(0,\zeta)} &\dfrac{(ad-bc)e^{i\phi}}{L(\zeta)\theta_1(0,\zeta)\theta_2(0,\zeta)}\\
\dfrac{e^{-i\phi}}{L(\zeta)\theta_1(0,\zeta)\theta_2(0,\zeta)} &-\dfrac{b\dfrac{\theta_1'(0,\zeta)}{\theta_1(0,\zeta)}-d}{L(\zeta)\theta_2^2(0,\zeta)}\\
\end{bmatrix}
\end{align*}
where $L(\zeta)=a\frac{\theta_1'(0,\zeta)}{\theta_1(0,\zeta)}-d\frac{\theta_2'(0,\zeta)}{\theta_2(0,\zeta)}+b\frac{\theta_1'(0,\zeta)\theta_2'(0,\zeta)}{\theta_1(0,\zeta)\theta_2(0,\zeta)}-c$.
\end{lemma} 
%
%
\begin{proof}
It is a well-known fact (see for example \cite{D1}), that the resolvent of the operator $H_{e_j}^D$ is an integral operator with symmetric kernel
\begin{align}\label{RD}
R_{e_j}^D(x,y;z)=\frac{\phi_j(x,\zeta)\theta_j(y,\zeta)}{w_j(\zeta)},\hspace{1cm}x\leq y,\hspace{1cm}\zeta=\sqrt z. 
\end{align}
Therefore the resolvent $R^D_V$ is a matrix integral operator with kernel
\begin{align}\label{Rinf}
R_{ j,l}^D(x,y;z)=\delta_{j,l}R_{j}^D(x,y;z)\hspace{1cm}j,l=1,2.
\end{align} 
For simplicity, we use the following notation
 $$R^{\mathcal A}_{j,l}(x,y;z)=R^{\mathcal A}_{j,l},\quad \phi_j(x,\zeta)=\phi_j(x),\quad \theta_j(y,\zeta)=\theta_j(y)\quad \mbox{and}\quad w_j(\zeta)=w_j.$$
By substituting (\ref{RD}) and (\ref{Rinf}) into \eqref{krein} and expressing it in matrix form we get
\begin{align*}
\begin{bmatrix}
R^{\mathcal A}_{1,1} &R^{\mathcal A}_{1,2}\\
R^{\mathcal A}_{2,1} &R^{\mathcal A}_{2,2}
\end{bmatrix}=
\begin{bmatrix}
\frac{\phi_1(x)\theta_1(y)+w_1\lambda_{1,1}\theta_1(x)\theta_1(y)}{w_1} &\lambda_{1,2}\theta_1(x)\theta_2(y) \\
\lambda_{2,1}\theta_2(x)\theta_1(y) &\frac{\phi_2(x)\theta_2(y)+w_2\lambda_{2,2}\theta_2(x)\theta_2(y)}{w_2}
\end{bmatrix}.
\end{align*}
Multiplying $f$ from right and using the definition of function $\psi$ the above expression becomes
\begin{align*}
\begin{bmatrix}
\psi_1(x)\\
\psi_2(x)
\end{bmatrix}=
\begin{bmatrix}
\int_0^\infty \{\frac{\phi_1(x)\theta_1(y)}{w_1}f_1(y)+\sum_{i=1}^2\lambda_{1,i}\theta_1(x)\theta_i(y)f_i(y)\}dy\\
\int_0^\infty\{\frac{\phi_2(x)\theta_2(y)}{w_2}f_2(y)+\sum_{i=1}^2\lambda_{2,i}\theta_2(x)\theta_i(y)f_i(y)\}dy
\end{bmatrix}
\end{align*} 
which is equivalent to the following system of equations
\begin{align*}
\psi_j(x)=\int\limits_0^\infty\left\lbrace\frac{\phi_j(x)\theta_j(y)}{w_j}f_j(y)+\sum_{i=1}^n\lambda_{j,i}\theta_j(x)\theta_i(y)f_i(y)\right\rbrace dy,\hspace{0.5cm}j=1,2.\\
\end{align*}
In order to determine $\lambda_{j,l}$ uniquely, we first express conditions \eqref{gencond} as a system of two equations and first  apply the sub-condition $$\psi_1(0)=e^{i\phi}(a\psi_2(0)+b\psi_2'(0))$$ on $\psi$ and Dirichlet condition on $\phi$. This yields the following expression
\begin{align}\label{wi}
&\int\limits_0^\infty\left\lbrace\sum_{i=1}^2\lambda_{1,i}\theta_1(0)\theta_i(y)f_i(y)\right\rbrace dy=\\\nonumber
&e^{i\phi}\left(a\int\limits_0^\infty\left\lbrace\sum_{i=1}^2\lambda_{2,i}\theta_2(0)\theta_i(y)f_i(y)\right\rbrace dy+b\int\limits_0^\infty\left\lbrace\frac{\theta_2(y)f_2(y)}{\theta_2(0)}+\sum_{i=1}^2\lambda_{2,i}\theta_2'(0)\theta_i(y)f_i(y)\right\rbrace dy\right).
\end{align}
Applying the second sub-condition 
$$\psi_1'(0)=e^{i\phi}(c\psi_2(0)+d\psi_2'(0))$$
 on $\psi$ and Dirichlet condition on $\phi$, we obtain
\begin{align}\label{wdi}
&\int\limits_0^\infty\left\lbrace\frac{\theta_1(y)f_1(y)}{\theta_1(0)}+\sum_{i=1}^2\lambda_{1,i}\theta_1'(0)\theta_i(y)f_i(y)\right\rbrace dy=\\\nonumber
&e^{i\phi}\left(c\int\limits_0^\infty\left\lbrace\sum_{i=1}^2\lambda_{2,i}\theta_2(0)\theta_i(y)f_i(y)\right\rbrace dy+d\int\limits_0^\infty\left\lbrace\frac{\theta_2(y)f_2(y)}{\theta_2(0)}+\sum_{i=1}^2\lambda_{2,i}\theta_2'(0)\theta_i(y)f_i(y)\right\rbrace dy\right).
\end{align} 
Comparing the coefficients of $f_i(y)\theta_i(y)$ in \eqref{wi} and \eqref{wdi}, we get the following set of equations
\begin{align*}
\theta_1(0)\lambda_{1,2}-e^{i\phi}\left(a\theta_2(0)\lambda_{2,2}-b\left(\frac{1}{\theta_2(0)}+\theta_2'(0)\lambda_{2.2}\right)\right)&=0\\
\theta_1(0)\lambda_{1,1}-e^{i\phi}\left(a\theta_2(0)\lambda_{2,1}+b\theta_2'(0)\lambda_{2,1}\right)&=0\\
\theta_1'(0)\lambda_{1,2}-e^{i\phi}\left(c\theta_2(0)\lambda_{2,2}-d\left(\frac{1}{\theta_2(0)}+\theta_2'(0)\lambda_{2.2}\right)\right)&=0\\
\frac{1}{\theta_1(0)}+\theta_1'(0)\lambda_{1,1}-e^{i\phi}\left(c\theta_2(0)\lambda_{2,1}+d\theta_2'(0)\lambda_{2,1}\right)&=0
\end{align*}
We obtain the values of $\lambda_{j,l}$ $j,l=1,2$ by solving the above set of equations.
\end{proof}

Using \eqref{Rinf} and Lemma \eqref{lambdavalues} we can express the resolvent kernel of $H_V^{\mathcal A}$ given by \eqref{krein} as 
\begin{align}\label{thmeq1}
\begin{bmatrix}
R^{\mathcal A}_{1,1} &R^{\mathcal A}_{1,2}\\
R^{\mathcal A}_{2,1} &R^{\mathcal A}_{2,2}\\
\end{bmatrix}=\begin{bmatrix}
R_{e_1}^D-\dfrac{a+b\dfrac{\theta_2'(0)}{\theta_2(0)}}{L(\zeta)\theta_1^2(0)}\theta_1^2(x) &\dfrac{(ad-bc)e^{i\phi}}{L(\zeta)\theta_1(0)\theta_2(0)}\theta_1(x)\theta_2(y)\\
\dfrac{e^{-i\phi}}{L(\zeta)\theta_1(0)\theta_2(0)}\theta_2(x)\theta_1(y) &R_{e_2}^D-\dfrac{b\dfrac{\theta_1'(0)}{\theta_1(0)}-d}{L(\zeta)\theta_2^2(0)}\theta_2^2(x)\\
\end{bmatrix}.
\end{align}
Here, for simplicity, we used the notation $R^{\mathcal A}_{i,j}(x,y;\zeta)=R^{\mathcal A}_{i,j}$, $R_{e_i}^D(x,y;\zeta)=R_{e_i}^D$ and $\theta_i(x,\zeta)=\theta_i(x)$.\\

For unperturbed operator i.e., when $V_j(x)\equiv0$, we may let 
\begin{align*}
\theta_j(x,\zeta)=e^{ix_j\zeta}\Rightarrow L(\zeta)=(a-d)i\zeta-\left(b\zeta^2+c\right).
\end{align*} 
Hence
\begin{align}\label{thmeq2}
\begin{bmatrix}
R^{\mathcal A}_{0,(1,1)} &R^{\mathcal A}_{0,(1,2)}\\
R^{\mathcal A}_{0,(2,1)} &R^{\mathcal A}_{0,(2,2)} 
\end{bmatrix}=
\begin{bmatrix}
R^D_{0,e_1}-\frac{a+i\zeta b}{(a-d)i\zeta-\left(b\zeta^2+c\right)}e^{2i\zeta x} &\frac{(ad-bc)}{(a-d)i\zeta-\left(b\zeta^2+c\right)}e^{i(x+y+\phi)\zeta}\\
\frac{1}{(a-d)i\zeta-\left(b\zeta^2+c\right)}e^{i(x+y-\phi)\zeta} &R^D_{0,e_2}-\frac{i\zeta b-d}{(a-d)i\zeta-\left(b\zeta^2+c\right)}e^{2i\zeta x}
\end{bmatrix}.
\end{align}
we now prove Theorem \eqref{t1}.
 \begin{proof}[Proof of Theorem \ref{t1}]
 The operator difference $R^D_V(z)-R^D_0(z)$ is a trace class operator and the second term on the right hand side of \eqref{krein} is a finite rank perturbation. Therefore,  the difference $R_V^{\mathcal A}(z)-R_0^{\mathcal A}$ is a trace class operator. The operator $H_V^{\mathcal A}$ can be seen as a particular case of a more general matrix valued Schr\"odinger operator considered in \cite{rec3} on the half-line when potential is a diagonal matrix. The fact that $R_V^{\mathcal A}(z)-R_0^{\mathcal A}$ is trace class also follows from Lemma 9.1 of \cite{rec3}. Representation \eqref{krein} with values of $\lambda_{j,l}$ given by lemma \eqref{lambdavalues} allows us to compute the trace of the resolvent difference of operators $H_V^{\mathcal A}$ and $H_0^{\mathcal A}$. \\

Using  \eqref{thmeq1} and \eqref{thmeq2} the trace of the difference of $R^{\mathcal A}_V$ and $R^{\mathcal A}_0$ can be written as
\begin{align}\nonumber
\Tr(R^{\mathcal A}_{V}(z)-R^{\mathcal A}_{0}(z))&=\sum\limits_{j=1}^2\int\limits_{0}^{\infty}\left(R^D_{e_j}(x,x,z)-R^D_{0,e_j}(x,x,z) \right)dx\\\nonumber
&-\int\limits_{0}^{\infty}\left(\dfrac{a+b\dfrac{\theta_2'(0)}{\theta_2(0)}}{L(\zeta)\theta_1^2(0)}\theta_1^2(x)+\dfrac{b\dfrac{\theta_1'(0)}{\theta_1(0)}-d}{L(\zeta)\theta_2^2(0)}\theta_2^2(x)\right)dx\\\label{t1meq}
&+\int\limits_{0}^{\infty}\left(\dfrac{(a-d+2i\zeta b)}{(a-d)i\zeta-(b\zeta^2+c)}e^{2ix\zeta} \right)dx.
\end{align}
The first term on the right hand side of \eqref{t1meq} is given by (c.f. Remark 1.2 of \cite{D1})
\begin{align}\label{t11}
\sum\limits_{j=1}^2\int\limits_{0}^{\infty}\left(R^D_{e_j}(x,x,z)-R^D_{0,e_j}(x,x,z) \right)dx=-\frac{1}{2\zeta}\sum\limits_{j=1}^2\frac{\dot{w}_j(\zeta)}{w_j(\zeta)}.
\end{align}
To compute the first term in the second integral on the right hand side of \eqref{t1meq} we use the following equation which is true for any arbitrary solutions $f_j(x,\zeta)$ and $g_j(x,\zeta)$ of $H^D_{e_j}\psi_j=\zeta^2\psi_j$
\begin{align}\label{idty}
f_j(x,\zeta)g_j(x,\zeta)=\dfrac{(f'_j(x,\zeta)\dot{g}_j(x,\zeta)-f_j(x,\zeta)\dot{g}'_j(x,\zeta))'}{2\zeta}.
\end{align}  
If we let $f_j=g_j=\theta_j$ we obtain
\begin{align*}
\int\limits_{0}^{\infty}\theta_j^2(x,\zeta)dx=\dfrac{[\theta'_j(x,\zeta)\dot{\theta}_j(x,\zeta)-\theta_j(x,\zeta)\dot{\theta}'_j(x,\zeta)]_0^{\infty}}{2\zeta}.
\end{align*}
If $V_j$ is compactly supported potential then the Jost solution $\theta_j(x,\zeta)=e^{i\zeta x}$ and hence
\begin{align*}
\theta'_j(x,\zeta)\dot{\theta}_j(x,\zeta)-\theta_j(x,\zeta)\dot{\theta}'_j(x,\zeta)=ie^{2i\zeta x}.
\end{align*}
For $\im \zeta>0$, $ie^{2i\zeta x}\rightarrow 0$ as $x\rightarrow\infty$. Therefore,  we are left with
\begin{align*}
&\int\limits_{0}^{\infty}\theta_j^2(x,\zeta)dx=\dfrac{\theta_j(0,\zeta)\dot{\theta}'_j(0,\zeta)-\theta'_j(0,\zeta)\dot{\theta}_j(0,\zeta)}{2\zeta}\\
&=\dfrac{\theta_j^2(0,\zeta)}{2\zeta}\dfrac{d}{d\zeta}\dfrac{\theta'_j(0,\zeta)}{\theta_j(0,\zeta)}.
\end{align*}
This implies 
\begin{align}\nonumber
&\int\limits_{0}^{\infty}\left(\dfrac{a+b\dfrac{\theta_2'(0)}{\theta_2(0)}}{L(\zeta)\theta_1^2(0)}\theta_1^2(x)+\dfrac{b\dfrac{\theta_1'(0)}{\theta_1(0)}-d}{L(\zeta)\theta_2^2(0)}\theta_2^2(x)\right)dx\\\nonumber
&=\dfrac{1}{2\zeta L(\zeta)}\left(\left(a+b\dfrac{\theta_2'(0)}{\theta_2(0)}\right)\dfrac{d}{d\zeta}\dfrac{\theta'_1(0,\zeta)}{\theta_1(0,\zeta)}+\left(b\dfrac{\theta_1'(0)}{\theta_1(0)}-d\right)\dfrac{d}{d\zeta}\dfrac{\theta'_2(0,\zeta)}{\theta_2(0,\zeta)}\right)\\\label{t12}
&=\dfrac{\dot{L}(\zeta)}{2\zeta L(\zeta)}.
\end{align}
By a density argument the above result can be extended to all potentials satisfying the condition $\int_0^\infty|V_j(x)|<\infty$.\\

Now it only remains to compute the last integral on the right hand side of \eqref{t1meq}, which is
\begin{align}\label{t13}
\int\limits_{0}^{\infty}\left(\dfrac{(a-d+2i\zeta b)}{(a-d)i\zeta-b\zeta^2+c}e^{2ix\zeta} \right)dx=\dfrac{(a-d+2i\zeta b)}{2\zeta((a-d)\zeta+(b\zeta^2+c)i)}.
\end{align}
Substituting values from  \eqref{t11}, \eqref{t12} and \eqref{t13} in \eqref{t1meq}, we get
\begin{align*}
\Tr(R^{\mathcal A}_{V}(z)-R^{\mathcal A}_{0}(z))&=-\frac{1}{2\zeta}\left(\sum\limits_{j=1}^2\frac{\dot{w}_j(\zeta)}{w_j(\zeta)}+\dfrac{\dot{L}(\zeta)}{L(\zeta)}-\dfrac{(a-d+2i\zeta b)}{(a-d)\zeta+(b\zeta^2+c)i}\right)\\
&=-\frac{1}{2\zeta}\left(\sum\limits_{j=1}^2\frac{d}{d\zeta}\ln({w_j(\zeta)})+\frac{d}{d\zeta}\ln({L(\zeta)})-\frac{d}{d\zeta}\ln({(a-d)\zeta+(b\zeta^2+c)i})\right)\\
&=-\frac{1}{2\zeta}\left(\dfrac{d}{d\zeta}\ln\left(\dfrac{L(\zeta)}{(a-d)\zeta+(b\zeta^2+c)i}\prod\limits_{j=1}^2w_j(\zeta)\right)\right).
\end{align*}

The second term on the right hand side of \eqref{krein} with values of $\lambda_{j,l}$ given by lemma \eqref{lambdavalues} is a finite rank operator and if the potential function $V$ satisfy condition \eqref{vcond} then the operator $\sqrt{|V|}(H^{\mathcal A}_0)^{-1/2}$ is Hilbert-Schmidt and therefore, $\sqrt{V}R_0^{\mathcal A}\sqrt{|V|}$ is trace class. This implies that the perturbation determinant $D(z)$ is well defined. The perturbation determinant and trace of the resolvent difference $R^{\mathcal A}_V(z)-R_0^{\mathcal A}(z)$ are related by the following expression
$$
\Tr(R^{\mathcal A}_V(z)-R_0^{\mathcal A}(z))=\frac{\frac{d}{dz}D(z)}{D(z)},\quad z\in\rho(H_V^{\mathcal A})\cap\rho(H_0^{\mathcal A}).
$$

%
%
In order to find an explicit expression for the perturbation determinant  we choose $\zeta=\sqrt{z}$ such that Im $(z)>0$. It follows from Theorem \ref{t1} that
\begin{align*}
\dfrac{\dfrac{d}{dz}( D(z))}{D(z)}=\frac{\dfrac{d}{dz}\left(\dfrac{L(\sqrt{z})}{(a-d)\sqrt{z}+(bz+c)i}\prod\limits_{j=1}^2w_j(\sqrt{z})\right)}{\dfrac{L(\sqrt{z})}{(a-d)\sqrt{z}+(bz+c)i}\prod\limits_{j=1}^2w_j(\sqrt{z})}.
\end{align*}
From here we deduce that
\begin{align*}
D(z)=A\dfrac{L(\sqrt{z})}{(a-d)\sqrt{z}+(bz+c)i}\prod\limits_{j=1}^2w_j(\sqrt{z}),\hspace{0.5cm}A\in\mathbb{C}.
\end{align*}
To find the value of the complex coefficient $A$, we use the following asymptotic of perturbation determinant
\begin{align*}
\lim_{|\text{Im}(z)|\rightarrow\infty}D(z)=1.
\end{align*}
This asymptotic holds if $\sqrt{|V|R^{\mathcal A}_0(z)^{-1}}$ is Hilbert-Schmidt operator (see \cite{yaf}). For the asymptotic behaviour of $w_j(\sqrt{z})$ as $|\sqrt{z}|\rightarrow\infty$ , we refer to \cite{DU}, which provides
\begin{align*}
w_j(\sqrt{z})=\theta_j(0,\sqrt{z})=1+O(|\sqrt{z}|^{-1}).
\end{align*}
Further, it is easy to find that
\begin{align}\label{kasymp}
L(\sqrt{z})=(a-d)\sqrt{z}i-(bz+c)+O(1),\quad |\sqrt{z}|\rightarrow \infty.
\end{align}
This implies, $A=\dfrac{1}{i}$. 
\end{proof}
%
%

\subsection{Spectral shift function}
In this section we briefly discuss the relationship between perturbation determinant $D(z)$ and the spectral shift function $\xi(\lambda; H_V^\mathcal A, H_0^\mathcal A)$.  The spectral shift function  for the pair of operators $H_V^{\mathcal A}$ and $H_0^{\mathcal A}$ can be defined by the trace formula
\begin{equation}\label{ssfdef}
\Tr(f(H_V^{\mathcal A}-H_0^{\mathcal A})):=\int_{-\infty}^{\infty}\xi(\lambda; H_V^{\mathcal A},H_0^{\mathcal A})f'(\lambda)\, d\lambda
\end{equation}
for all $f\in C_0^\infty(\mathbb R)$.\\

Theorem 9.2 of \cite{rec3} explicitly characterize the spectral shift function in the more general case of matrix valued Schr\"odinger operators with general self-adjoint conditions at the origin. It was proved that the trace formula \eqref{ssfdef} holds for a bigger class of functions satisfying $f'(\lambda)=O(\lambda^{-1/2-\epsilon}),\, f''(\lambda)=O(\lambda^{-1-\epsilon}),\,\epsilon>0,\, \lambda\rightarrow \infty$ and 
$$
\int_{-\infty}^\infty|\xi(\lambda;H^{\mathcal A}_V,H^{\mathcal A}_{0})|(1+|\lambda|)^{-1/2-\epsilon}\,d\lambda<\infty,\quad \epsilon >0.
$$
One of the main ingredients in the proof of Theorem 9.2 of \cite{rec3} is the following estimate\footnote[1]{The interested reader can consult Appendix \eqref{appendix} where explicit derivation of this  estimate is provided.} (cf. inequality (9.2) of \cite{rec3}, see also (5.11) of \cite{yaf} )
\begin{equation}\label{resest}
||R^{\mathcal A}_{V}(-t)-R^{\mathcal A}_{0}(-t)||_1\leq C_\epsilon t^{-\frac{3}{2}+\epsilon},\quad \epsilon >0,\quad t\rightarrow \infty.
\end{equation}
The above estimate implies
$$
\int_1^\infty t^{-m}||R^{\mathcal A}_{V}(-t)-R^{\mathcal A}_{0}(-t)||_1<\infty $$
which holds for all $m>-1/2$. Proposition (3.3) of \cite{D1} then implies
$$
\xi(\lambda;\,H^{\mathcal A}_V,H^{\mathcal A}_0)=\pi^{-1}\lim_{\epsilon\rightarrow 0}\emph{arg}\,D(\lambda+i\epsilon),
$$
where $\emph{arg}\,D(z)=\im\ln D(z)$ is defined by the condition $\ln D(z)\rightarrow 0$ as $\dist\{z,\sigma(H^{\mathcal{A}}_0)\}\rightarrow \infty$. Moreover,
$$
\ln D(z)=\int_{-\infty}^\infty \xi(\lambda;\,H^{\mathcal A}_V,H^{\mathcal A}_0)(\lambda-z)^{-1}\,d\lambda, \quad z\in\rho(H^{\mathcal A}_0)\cap\rho(H^{\mathcal A}_V).
$$


\section{Levinson's formula}
In this section we derive a zero order trace formula, commonly known as Levinson's formula, which describes the number of negative eigenvalues of an operator in terms of phase shift function which is denoted by $\eta$ and it is defined, for real $k$, as $\eta(k):=\arg D(k)$. The Levinson's formula derived in this section is closely related to the Levinson's formulas obtained in Theorem 9.3 of \cite{Ak2} and Theorem 9.3 of \cite{rec3}. Levinson's formula obtained in \cite{Ak2} relates the number of negative eigenvalues with the determinant of the scattering matrix and certain parameters that depend on the boundary conditions. Formula obtained in \cite{rec3} relates the number of negative eigenvalues with the spectral shift function. The vertex conditions of this paper can be  expressed as a particular case of more general conditions considered in references \cite{Ak2} and \cite{rec3} and therefore the Levinson's formulas obtained there also hold for the operator $H_V^\mathcal A$.  For higher order trace formulas of integer and half-integer order we refer to \cite{yaf, rec2, UA}.\\

We will need the following low-energy asymptotics of the perturbation determinant\begin{lemma}[Low energy asymptotics]\label{Low_ergy_Asymp}
Let the potential $V$ satisfies the condition
\begin{align}\label{Dec_Condition}
\int\limits_{e_j}(1+x)|V_j(x)|dx<\infty,\hspace*{0.3cm}1\leq j\leq n.
\end{align}
Then as $\zeta\rightarrow0$ the perturbation determinant satisfies

\[ D(\zeta)=\begin{cases} 
      \zeta^{(Q-P)}(1+o(1)) & w_j(0)\neq0$, $j=1,2,\\
      \zeta^{(R-P)}(1+o(1)) & w_1(0)=0,w_2(0)\neq0,\\
             \zeta^{(S-P)}(1+o(1)) & w_1(0)\neq0,w_2(0)=0,\\
                   \zeta^{(T-P)}(1+o(1)) & w_j(0)=0$, $j=1,2.
   \end{cases}
\]
Where the functions $P,Q,R, S$ and $T$ are defined in Theorem \eqref{lea}.
\end{lemma}

\begin{proof}
Expression for perturbation determinant is
\begin{align*}
D(\zeta)=\dfrac{L(\zeta)\prod\limits_{j=1}^2w_j(\zeta)}{(a-d)i\zeta-(b\zeta^2+c)}.
\end{align*}

We consider the following five different cases
\begin{enumerate}
\item $w_j(0)\neq0$, $j=1,2$ and $L(0)\neq 0$. 
\item $w_j(0)\neq 0$, $j=1,2$ and $L(0)= 0$.
\item $w_1(0)=0$ and $w_2(0)\neq0$.
\item $w_1(0)\neq 0$ and $w_2(0)=0$.
\item $w_j(0)=0$, $j=1,2$.
\end{enumerate}
In the first case $D(\zeta)$ has no zeros and we can let
\begin{align*}
&L(0)\prod\limits_{j=1}^2w_j(0)=\tilde{c}\neq0\\
&\Rightarrow \lim\limits_{\zeta\rightarrow0}D(\zeta)=\lim\limits_{\zeta\rightarrow0}\dfrac{\tilde{c}}{(a-d)i\zeta-(b\zeta^2+c)}\\
&\Rightarrow \lim\limits_{\zeta\rightarrow0}D(\zeta)((a-d)i\zeta-(b\zeta^2+c))=\tilde{c}\\
&\Rightarrow \lim\limits_{\zeta\rightarrow0}\dfrac{D(\zeta)((a-d)i\zeta-(b\zeta^2+c))-\tilde{c}}{\tilde{c}}=0\\
&\Rightarrow D(\zeta)=\frac{\tilde{c}}{(a-d)i\zeta-(b\zeta^2+c)}(1+o(1)).
\end{align*}
In the second case, we make use of the following small energy asymptotics (c.f. Theorem 2.3 (iii)of \cite{Ak1})  
\begin{align*}
\dfrac{\theta_j'(0,\zeta)}{\theta_j(0,\zeta)}=\dfrac{\theta_j'(0,0)}{\theta_j(0,0)}-\dfrac{i\zeta}{\theta_j^2(0,0)}+o(\zeta)
\end{align*}
which implies
\begin{align*}
L(\zeta)=&a\left(\frac{\theta_1'(0,0)}{\theta_1(0,0)}-\frac{i\zeta}{\theta_1^2(0,0)}+o(\zeta)\right)-d\left(\frac{\theta_2'(0,0)}{\theta_2(0,0)}-\frac{i\zeta}{\theta_2^2(0,0)}+o(\zeta)\right)\\
&+b\left(\frac{\theta_1'(0,0)}{\theta_1(0,0)}-\frac{i\zeta}{\theta_1^2(0,0)}+o(\zeta)\right)\left(\frac{\theta_2'(0,0)}{\theta_2(0,0)}-\frac{i\zeta}{\theta_2^2(0,0)}+o(\zeta)\right)+c\\
=&a\frac{\theta_1'(0,0)}{\theta_1(0,0)}-d\frac{\theta_2'(0,0)}{\theta_2(0,0)}+b\frac{\theta_1'(0,0)}{\theta_1(0,0)}\frac{\theta_2'(0,0)}{\theta_2(0,0)}-c+a\left(\frac{-i\zeta}{\theta_1^2(0,0)}+o(\zeta)\right)\\
&-d\left(\frac{-i\zeta}{\theta_2^2(0,0)}+o(\zeta)\right)+b\left[\left(\frac{\theta_1'(0,0)}{\theta_1(0,0)\theta_2^2(0,0)}+\frac{\theta_2'(0,0)}{\theta_1^2(0,0)\theta_2(0,0)}\right)i\zeta\right]\\
&+b\left(o(\zeta)+o(\zeta^2)-\frac{\zeta^2}{\theta_1(0,0)\theta_2(0,0)}\right)
\end{align*}
This yields
\begin{align*}
L(\zeta)=&L(0)+a\left(\frac{-i\zeta}{\theta_1^2(0,0)}+o(\zeta)\right)-d\left(\frac{-i\zeta}{\theta_2^2(0,0)}+o(\zeta)\right)\\
&+b\left[\left(\frac{\theta_1'(0,0)}{\theta_1(0,0)\theta_2^2(0,0)}+\frac{\theta_2'(0,0)}{\theta_1^2(0,0)\theta_2(0,0)}\right)i\zeta+o(\zeta)+o(\zeta^2)-\frac{\zeta^2}{\theta_1(0,0)\theta_2(0,0)}\right].
\end{align*}
Therefore,
\begin{align*}
D(\zeta)=\dfrac{\alpha_1\zeta+o(\zeta)+b(\zeta^2+(o(\zeta^2)))}{(a-d)i\zeta-(b\zeta^2+c)},
\end{align*}
where, $\alpha_1=\left(b\left(\dfrac{\theta_1'(0,0)}{\theta_2(0,0)}+\dfrac{\theta_2'(0,0)}{\theta_1(0,0)}\right)-a\dfrac{\theta_2(0,0)}{\theta_1(0,0)}+d\dfrac{\theta_1(0,0)}{\theta_2(0,0)}\right)i$.\\
The above calculations show that the order of the numerator of $D(\zeta)$ depends on the value of the parameter $b$. If $b=0$, then it will be of order $\zeta$ otherwise it will be of order $\zeta^2$.
 
The study of first two cases can be summarized as follows
\begin{itemize}
\item If $w_j(0)\neq 0$, $j=1,2$ and $L(0)\neq 0$ then the numerator of $D(\zeta)$ is a non zero constant and it will have no zeros.
\item  If $w_j(0)\neq 0$, $j=1,2$ and $L(0)=0$ then the numerator of $D(\zeta)$ for small $\zeta$ will have one zero in case $\alpha_1 \neq 0$ and it will have two zeros in case $\alpha_1 =0$, $b \neq 0$.
\end{itemize}
This gives us the following low energy asymptotics in the first two cases
\begin{align*}
D(\zeta)=\zeta^{(Q-P)}(1+o(1)),
\end{align*}
where
\[ Q=Q(b,\alpha_1,L(0))=\begin{cases} 
      0 & L(0)\neq0, \\
      1 & L(0)=0,\alpha_1\neq0\\
      2 & L(0)=0,\alpha_1=0,b\neq0 
   \end{cases}
\]
and
\[P= P(a,b,c,d)=\begin{cases} 
      0 & c\neq0, \\
      1 & c=0,a-d\neq0\\
      2 & c=0,a-d=0,b\neq0 .
   \end{cases}
\]
The function $P$ gives the number of poles of $D(\zeta)$ in the given cases.\\

Let us now consider the third case, i.e., when $w_1(0)=0$ and $w_2(0)\neq0$. We rewrite $D(\zeta)$ as
\begin{align}\nonumber
D(\zeta)&=\dfrac{\left(a\dfrac{\theta_1'(0,\zeta)}{\theta_1(0,\zeta)}-d\dfrac{\theta_2'(0,\zeta)}{\theta_2(0,\zeta)}+b\prod\limits_{k=1}^2\dfrac{\theta_k'(0,\zeta)}{\theta_k(0,\zeta)}-c\right)\prod\limits_{j=1}^2\theta_j(0,\zeta)}{(a-d)i\zeta-(b\zeta^2+c)}\\\label{sympD}
&=\dfrac{\left(a\theta_1'(0,\zeta)\theta_2(0,\zeta)+b\theta_1'(0,\zeta)\theta_2'(0,\zeta)-d\theta_1(0,\zeta)\theta_2'(0,\zeta)-c\prod\limits_{j=1}^2\theta_j(0,\zeta)\right)}{(a-d)i\zeta-(b\zeta^2+c)}
\end{align}
and use the following asymptotic for $\theta_1(x,\zeta)$ (c.f. \cite{dei})  \begin{align*}
 \theta_1(0,\zeta)=\zeta\dot{\theta}_1(0,0)+o(\zeta).
 \end{align*}
 This implies
 \begin{align*}
 D(\zeta)=\dfrac{\left(\alpha_2-c_3\zeta+o(\zeta)\right)}{(a-d)i\zeta-(b\zeta^2+c)}
 \end{align*}
Where $\alpha_2=a\theta_1'(0,0)\theta_2(0,0)+b\theta_1'(0,0)\theta_2'(0,0)$ and $c_3=(c\theta_2(0,0)+d\theta_2'(0,0))\dot{\theta}_1(0,0)$.
The numerator of $D(\zeta)$ is equal to $\alpha_2-c_3\zeta+o(\zeta)$ as $\zeta \rightarrow 0$.
If $\alpha_2\neq0$ then the numerator of $D(\zeta)$ has no zeros for small $\zeta$ and  if $\alpha_2=0$, then the numerator of $D(\zeta)$ has one zero for small $\zeta$. Hence,
\begin{align*}
D(\zeta)=\zeta^{(R-P)}(1+o(1)),\quad \zeta \rightarrow 0,
\end{align*}
where,
\[ R=R(\alpha_2)=\begin{cases} 
      0 & \alpha_2\neq0, \\
      1 & \alpha_2=0.
   \end{cases}
\]\\

For the fourth case we substitute
\begin{align*}
 \theta_2(0,\zeta)=\zeta\dot{\theta}_2(0,0)+o(\zeta),\zeta \rightarrow 0
\end{align*}
into \eqref{sympD} and obtain
 \begin{align*}
 D(\zeta)=\dfrac{\left(\alpha_3+c_5\zeta+o(\zeta)\right)}{(a-d)i\zeta-(b\zeta^2+c)}.
 \end{align*}
Here, $\alpha_3=b\theta_1'(0,0)\theta_2'(0,0)-d\theta_1(0,0)\theta_2'(0,0)$ and $c_5=(a\theta_1'(0,0)-c\theta_1(0,0))\dot{\theta}_2(0,0)$. The numerator of $D(\zeta)$ for small $\zeta$ has no zeros if $\alpha_3 \neq 0$ and it has one zero if $\alpha_3 =0$.
Therefore,
\begin{align*}
D(\zeta)=\zeta^{(S-P)}(1+o(1)),\quad \zeta\rightarrow 0,
\end{align*}
where, 
\[S=S(\alpha_3)=\begin{cases} 
      0 & \alpha_3\neq0, \\
      1 & \alpha_3=0.
   \end{cases}
\]\\

Finally, for the fifth case we substitute the asymptotics
\begin{align*}
\theta_1(0,\zeta)=\zeta\dot{\theta}_1(0,0)+o(\zeta) \text{ and }\theta_2(0,\zeta)=\zeta\dot{\theta}_2(0,0)+o(\zeta),\quad\zeta \rightarrow 0
\end{align*}
into \eqref{sympD} and obtain
 \begin{align*}
 D(\zeta)=\dfrac{c\left(c_6\zeta^2+o(\zeta^2)\right)+\alpha_4\zeta+bc_8+o(\zeta)}{(a-d)i\zeta-(b\zeta^2+c)},\quad \zeta \rightarrow 0.
 \end{align*}
Here, $c_6=\dot{\theta}_1(0,0)\dot{\theta}_2(0,0)\neq0$, $c_8=\theta_1'(0,0)\theta_2'(0,0)\neq0$ (because $\theta(0,0)=0\Rightarrow\dot{\theta}(0,0)\theta'(0,0)=-i$, the proof is given in \cite[~Lemma 4.6.(1)]{D1}, hence $\dot{\theta}(0,0)\neq0$ and $\theta'(0,0)\neq0$) $\alpha_4=a\theta_1'(0,0)\dot{\theta}_2(0,0)-d\dot{\theta}_1(0,0)\theta_2'(0,0)$.
This implies the numerator of $D(\zeta)$ for small $\zeta$ has no zeros if $b\neq 0$, it has one zero if $b=0, \alpha_4\neq0$ and it has two zeros if $b=0$, $\alpha_4=0$, $c\neq 0$.
This yields
\begin{align*}
D(\zeta)=\zeta^{(T-P)}(1+o(1)),\quad \zeta\rightarrow 0,
\end{align*}
where
\[T=T(b,c,\alpha_4)=\begin{cases} 
      0 & b\neq0, \\
      1 & b=0,\alpha_4\neq0\\
      2 & b=0,\alpha_4=0,c\neq0 .
   \end{cases}
\]

\end{proof}
Using the low energy asymptotics of $D(\zeta)$ we now prove the analogue of Levinson's formula for the negative eigenvalues of $H_V^{\mathcal A}$, stated in Theorem \eqref{lea}.
\begin{proof}[Proof of Theorem \ref{lea}]
The function $D(\zeta)$ has a zero in $\zeta$ of order $r$ if and only if $\zeta^2$ is an eigenvalue of multiplicity $r$ of the operator $H_{V}^{\mathcal A}$ \cite{Ak2}. As $H_V^{\mathcal A}$ is a self-adjoint operator, so the zeros of $D(\zeta)$ lie on the positive imaginary axis and it may have only real eigenvalues.\\

Let $\Gamma_{R,\epsilon}$ denotes the contour (counterclockwise) consisting of semicircles $C^+_R=\{|\zeta|=R,\quad \mbox{Im } \zeta\geq0\}$ and $C^+_\epsilon=\{|\zeta|=\epsilon,\quad \mbox{Im } \zeta\geq0\}$ and the intervals $(\epsilon,R)$ and $(-R,-\epsilon)$.
$R$ and $\epsilon$ are chosen such that all $N$ negative eigenvalues lie inside the contour $\Gamma_{R,\epsilon}$.\\

The function $w_j(\zeta)$ is analytic in the upper halfp-plane, therefore, $D(\zeta)$ is analytic inside and on the contour $\Gamma_{R,\epsilon}$. Hence,
\begin{align}\label{arg_prin}
\int\limits_{\Gamma_{R,\epsilon}}\dfrac{\frac{d}{d\zeta}D(\zeta)}{D(\zeta)}d\zeta=2\pi iN.
\end{align}
Note that,
\begin{align}\label{limit_D(xi)}
\lim\limits_{\mbox{Im }\zeta\rightarrow\infty}D(\zeta)=\lim\limits_{\mbox{Im }\zeta\rightarrow\infty}\dfrac{L(\zeta)\prod\limits_{j=1}^2w_j(\zeta)}{(a-d)i\zeta-(b\zeta^2+c)}=1+O(|\zeta|^{-1}).
\end{align}
The branch of the function $\ln D(\zeta)$ can be fixed by the condition $\ln D(\zeta)\rightarrow0$ as $|\zeta|\rightarrow\infty$. 
For $k\in \mathbb{R}$, we set $D(k)=a(k)e^{i \eta (k)}$, where $a(k)=|D(k)|$ and $\eta(k)=\arg D(k)$. 
Now
\begin{align*}
\mbox{var}_{\Gamma_{R,\epsilon}}\arg D(\zeta)=\eta(R)-\eta(\epsilon)+\mbox{var}_{C_R^+}\arg D(\zeta)+\{\eta(-\epsilon)-\eta(-R)\}+\mbox{var}_{C_{\epsilon}^+}\arg D(\zeta).
\end{align*}
It follows from the representation of $D(k)$ that $\eta(-k)=-\eta(k)$. Hence,
\begin{align*}
\mbox{var}_{\Gamma_{R,\epsilon}}\arg D(\zeta)=2(\eta(R)-\eta(\epsilon))+\mbox{var}_{C_R^+}\arg D(\zeta)+\mbox{var}_{C_{\epsilon}^+}\arg D(\zeta)
\end{align*}
Equation \eqref{arg_prin} implies that
\begin{align*}
\mbox{var}_{\Gamma_{R,\epsilon}}\arg D(\zeta)=2\pi N.
\end{align*}
This implies
\begin{align*}
2\pi N=2(\eta(R)-\eta(\epsilon))+\mbox{var}_{C_R^+}\arg D(\zeta)+\mbox{var}_{C_{\epsilon}^+}\arg D(\zeta).
\end{align*}
Using \eqref{limit_D(xi)}, one can see that $\lim\limits_{R\rightarrow\infty}\mbox{var}_{C_R^+}\arg D(\zeta)=0$. We are left with
\begin{align*}
\eta(\infty)-\eta(\epsilon)=\pi N-\dfrac{1}{2}\mbox{var}_{C_{\epsilon}^+}\arg D(\zeta).
\end{align*}
Now letting $\epsilon\rightarrow0$ and using Lemma \eqref{Low_ergy_Asymp}, we arrive at
\[ \lim\limits_{\epsilon\rightarrow0}\mbox{var}_{C_{\epsilon}^+}\arg D(\zeta)=\begin{cases} 
      -(Q-P)\pi  & w_j(0)\neq0$, $j=1,2,\\
      -(R-P)\pi  & w_1(0)=0,w_2(0)\neq0,\\
      -(S-P)\pi  & w_1(0)\neq0,w_2(0)=0,\\
      -(T-P)\pi  & w_j(0)=0$, $j=1,2.
   \end{cases}
\]
Therefore,
\[ \eta(\infty)-\eta(0)=\begin{cases} 
      \pi (N+\dfrac{Q-P}{2}) & w_j(0)\neq0$, $j=1,2,\\
      \pi (N+\dfrac{R-P}{2}) & w_1(0)=0,w_2(0)\neq0,\\
      \pi (N+\dfrac{S-P}{2}) & w_1(0)\neq0,w_2(0)=0,\\
      \pi (N+\dfrac{T-P}{2}) & w_j(0)=0$, $j=1,2.
   \end{cases}
\]
\end{proof} 
\section{examples}
In this section we illustrate the trace formula and perturbation determinant obtained in theorem \eqref{t1} with some examples. 
\begin{example} For the $\delta$ interaction of strength $\alpha\in\mathbb R$, we can choose $\phi=0,\, a=1,\, b=0,\, c=\alpha$ and $d=-1$. Then according to theorem \eqref{t1} the trace of the resolvent difference is given by
$$
\Tr(R_V^{\mathcal A}(\zeta)-R_0^{\mathcal A}(\zeta))=-\frac{1}{2\zeta}\ln \left(\frac{\theta_1'(0,\zeta)\theta_2(0,\zeta)+\theta_1(0,\zeta)\theta_2'(0,\zeta)-\alpha\theta_1(0,\zeta)\theta_2(0,\zeta)}{2\zeta+i\alpha }\right)
$$
and the expression for perturbation determinant is
$$
D(\zeta)=\frac{1}{2i\zeta-\alpha}\left(\theta_1'(0,\zeta)\theta_2(0,\zeta)+\theta_1(0,\zeta)\theta_2'(0,\zeta)-\alpha\theta_1(0,\zeta)\theta_2(0,\zeta)\right).
$$
For $\alpha=0$ the $\delta$ conditions are commonly known as Kirchhoff conditions and in this case the above expressions coincide with the results obtained in Theorem (1.1) and Corollary (3.1) of \cite{D1} with $n=2$ (see also \cite{UD}).
\end{example}

\begin{example}
The $\delta'$ conditions are obtained by choosing $\phi=0,\, a=-1,\,b=\beta,\,c=0$ and $d=1$ in \eqref{gencond}. In this case theorem \eqref{t1} implies
$$
\Tr(R_V^{\mathcal A}(\zeta)-R_0^{\mathcal A}(\zeta))=-\frac{1}{2\zeta}\ln \left(\frac{\theta_1'(0,\zeta)\theta_2(0,\zeta)+\theta_1(0,\zeta)\theta_2'(0,\zeta)-\beta\theta'_1(0,\zeta)\theta'_2(0,\zeta)}{2\zeta-i\beta\zeta^2 }\right)
$$

and perturbation determinant is given by
$$
D(\zeta)=\frac{1}{2i\zeta+\beta\zeta^2}\left(\theta_1'(0,\zeta)\theta_2(0,\zeta)+\theta_1(0,\zeta)\theta_2'(0,\zeta)-\beta\theta'_1(0,\zeta)\theta'_2(0,\zeta)\right).
$$
\end{example}

In the next two examples we use the notation $D_x\phi=\phi^{(1)}$ for the generalized derivative in the sense of \cite{kurdis}.
\begin{example}
The Schr\"odinger operator with the singular density $-D_x\left(1+\sigma\delta\right)D_x$  is equivalent to the operator $H_V^{\mathcal A}$ with $\phi=0,\, a=1, \, b=-\sigma\in \mathbb R,\, c=0$ and $d=-1$ and the perturbation determinant is given by
$$
D(\zeta)=\frac{1}{2i\zeta +\sigma\zeta^2 }\left(\theta_1'(0,\zeta)\theta_2(0,\zeta)+\theta_1(0,\zeta)\theta'_2(0,\zeta)-\sigma\theta'_1(0,\zeta)\theta_2'(0,\zeta)\right).
$$
\end{example}

\begin{example}
For real $\sigma_1$ and $\sigma_2$ the two parameters family of Schr\"odinger operators with singular potential $-D_x^2+\sigma_1\delta+\sigma_2\delta^{(1)}+V$ is equivalent to the operator $H_V^{\mathcal A}$ with $\phi=0,\, a=\frac{2+\sigma_2}{2-\sigma_2},\, b=0,\, c=\frac{4\sigma_1}{4-\sigma_2^2},\, d=-\left(\frac{2-\sigma_2}{2+\sigma_2}\right)$ (c.f. Section 3.2.4 of \cite{kurdis}). Theorem \eqref{t1} implies
the following expression for the perturbation determinant
$$
D(\zeta)=\dfrac{\left(\left(\frac{2+\sigma_2}{2-\sigma_2}\right)\theta_1'(0,\zeta)\theta_2(0,\zeta)+\left(\frac{2-\sigma_2}{2+\sigma_2}\right)\theta_1(0,\zeta)\theta_2'(0,\zeta)-\left(\frac{4\sigma_1}{4-\sigma_2^2}\right)\theta_1(0,\zeta)\theta_2(0,\zeta)\right).}{\left(\frac{8+2\sigma_2^2}{4-\sigma_2^2}\right)i\zeta-\left(\frac{4\sigma_1}{4-\sigma_2^2}\right)}
$$
\end{example}
\section{conclusions}
The Schr\"odinger operators on the real-line with generalized point interaction at zero can be imagined as a $2\times 2$ matrix-valued Schr\"odinger operator on semi-axis with general self-adjoint conditions at zero. 
Scattering theory of matrix-valued Schr\"odinger operators
on semi-axis with general self-adjoint conditions at zero has been studied in much detail
in the last few years.  On the other hand, explicit derivation of perturbation determinants and its connections with spectral shift function and Levinson's formulas  has received relatively less attention. The present article is devoted to explicit
derivation of perturbation determinant via trace formula for Schr\"odinger operators and expressing Levinson's formula and spectral shift function in terms of perturbation determinant. Our results complement the previous studies in this regard. We used operator theoretic approach in our derivations similar to Demirel's work which deals
with star graphs with Kirchhoff conditions at the vertex. The
expressions for perturbation determinants are further used to derive Levinsons formulas connecting
the number of negative eigenvalues with phase shift. Phase shift is defined in terms of perturbation
determinant. Looking forward,  it would be interesting to extend and generalize the results on perturbation determinant to
matrix-valued Schr\"odinger operators\ with most general self-adjoint conditions at zero.

\appendix
\section{Estimate on trace norm of resolvent difference}\label{appendix}
Here we'll provide and explicit derivation of the estimate
\begin{align*}
||R_{V}^{\mathcal A}(-t)-R_{0}^{\mathcal A}(-t)||_1\leq C_\epsilon t^{-\frac{3}{2}+\epsilon}
\end{align*}
for all $\epsilon>0$ and large $t$. We will need the following lemma.
\begin{lemma}[\cite{D1}, Lemma 3.4 ]\label{ssflemma1}
Let $\mathcal{H}$ be a Hilbert space and $f, g \in \mathcal{H}$. Assume that $\mathcal{R}=(\cdot,\,f)f-(\cdot,\,g)g$
is an operator of rank two on $\mathcal{H}$. Then, the trace norm of $\mathcal{R}$ is given by
\begin{align}\label{ssf1}
||\mathcal{R}||_1=\sqrt{(||f||^2+||g||^2)^2-4|(f,g)|^2}.
\end{align}
Moreover, if we let $g=f+h$ then,
\begin{align}\label{ssf2}
(||f||^2+||g||^2)^2-4|(f,g)|^2\leq6||h||^2||f||^2+3||h||^4.
\end{align}
\end{lemma}
\begin{lemma}\label{ssflemma2}
Assume that the potential $V$ satisfies condition \eqref{vcond}. Then the resolvents $R_V^{\mathcal A}$ of $H_V^{\mathcal A}$ and $R_0^{\mathcal A}$ of $H_0^{\mathcal A}$ satisfy, for all $\epsilon >0$ and for large $t,$
\begin{align*}
||R_{V}^{\mathcal A}(-t)-R_{0}^{\mathcal A}(-t)||_1\leq C_\epsilon t^{-\frac{3}{2}+\epsilon}.
\end{align*}
\end{lemma}
\begin{proof}
Let $\mathcal{R}_1=R^{\mathcal{A}}_{V}(-t)-R^D_{V}(-t)+R^D_{0}(-t)-R^{\mathcal{A}}_{0}(-t)$ and $\mathcal{R}_2=R^D_{V}(-t)-R^D_{0}(-t)$ then,
\begin{align*}
R^{\mathcal{A}}_{V}(-t)-R^{\mathcal{A}}_{0}(-t)=\mathcal{R}_1+\mathcal{R}_2.
\end{align*}
Hence
\begin{align}\label{trnormineq}
||R^{\mathcal{A}}_{V}(-t)-R^{\mathcal{A}}_{0}(-t)||_1\leq||\mathcal{R}_1||_1+||\mathcal{R}_2||_1.
\end{align}
Trace norm on $\mathcal{R}_2$ can be estimated as $||\mathcal{R}_2||_1\leq ct^{-\frac{3}{2}+\epsilon}$ (c.f. Lemma 4.5.6. \cite{yaf}). We only need to find estimate on the trace norm of rank two operator $\mathcal{R}_1$. By adding and subtracting the quantity
\begin{align*}
\dfrac{\theta_j^2(x_j,i\sqrt{t})}{L(i\sqrt{t})\theta_j^2(0,i\sqrt{t})}+\dfrac{e^{-2 x_j\sqrt{t}}}{(d-a)\sqrt{t}+bt-c}\hspace*{1cm}j=1,2
\end{align*}
in the diagonal of $\mathcal{R}_1$ we can re-write it as 
\begin{align*}
\mathcal{R}_1=\mathcal{S}_1+\mathcal{S}_2,
\end{align*} where
\begin{align*}
\mathcal{S}_1=\begin{bmatrix}
\frac{(a-\sqrt{t} b+1)e^{-2 x_1\sqrt{t}}}{(d-a)\sqrt{t}+bt-c}-\frac{\left(a+b\frac{\theta_2'(0,i\sqrt{t})}{\theta_2(0,i\sqrt{t})}+1\right)\theta_1^2(x_1,i\sqrt{t})}{L(i\sqrt{t})\theta_1^2(0,i\sqrt{t})} &0\\
0 &\frac{(1-\sqrt{t} b-d)e^{-2 x_2\sqrt{t}}}{(d-a)\sqrt{t}+bt-c}-\frac{\left(b\frac{\theta_1'(0,i\sqrt{t})}{\theta_1(0,i\sqrt{t})}-d+1\right)\theta_2^2(x_2,i\sqrt{t})}{L(i\sqrt{t})\theta_2^2(0,i\sqrt{t})}
\end{bmatrix}
\end{align*}
and
\begin{align*}
(\mathcal{S}_2)_{j,k}=\begin{bmatrix}
\frac{\theta_j(x_j,i\sqrt{t})\theta_k(x_k,i\sqrt{t})}{L(i\sqrt{t})\theta_j(0,i\sqrt{t})\theta_k(0,i\sqrt{t})}-\frac{e^{-(x_j+x_k)\sqrt{t}}}{(d-a)\sqrt{t}+bt-c}
\end{bmatrix}\hspace*{1cm}j,k=1,2.
\end{align*}
Now, as $\mathcal{S}_1$ is a diagonal matrix and in order to apply definition of trace norm, we can let $\mathcal{S}_1=I\mathcal{S}_1I$, where $I$ is second order identity matrix. This implies,
\begin{align*}
||\mathcal{S}_1||_1=|\sigma_1|+|\sigma_2|
\end{align*}
where
\begin{align*}
|\sigma_1|=\int_0^{\infty}\left|\frac{(a-\sqrt{t} b+1)e^{-2 x_1\sqrt{t}}}{(d-a)\sqrt{t}+bt-c}-\frac{\left(a+b\frac{\theta_2'(0,i\sqrt{t})}{\theta_2(0,i\sqrt{t})}+1\right)\theta_1^2(x_1,i\sqrt{t})}{L(i\sqrt{t})\theta_1^2(0,i\sqrt{t})}\right|dx_1
\end{align*}
and
\begin{align*}
|\sigma_2|=\int_0^{\infty}\left|\frac{(1-\sqrt{t} b-d)e^{-2 x_2\sqrt{t}}}{(d-a)\sqrt{t}+bt-c}-\frac{\left(b\frac{\theta_1'(0,i\sqrt{t})}{\theta_1(0,i\sqrt{t})}-d+1\right)\theta_2^2(x_2,i\sqrt{t})}{L(i\sqrt{t})\theta_2^2(0,i\sqrt{t})}\right|dx_2.
\end{align*}
For $t$ large enough, these expressions can be simplified to
\begin{align*}
|\sigma_1|&\approx\left|\frac{(a-\sqrt{t} b+1)}{(d-a)\sqrt{t}+bt-c}\right|\int_0^{\infty}\left|e^{-2 x_1\sqrt{t}}-\theta_1^2(x_1,i\sqrt{t})\right|dx_1\\
&=\left|\frac{(a-\sqrt{t} b+1)}{(d-a)\sqrt{t}+bt-c}\right|\int_0^{\infty}\left|\theta_1(x_1,i\sqrt{t})-e^{-x_1\sqrt{t}}\right|\left|\theta_1(x_1,i\sqrt{t})+e^{- x_1\sqrt{t}}\right|dx_1
\end{align*}
and
\begin{align*}
|\sigma_2|\approx\left|\frac{(1-\sqrt{t} b-d)}{(d-a)\sqrt{t}+bt-c}\right|\int_0^{\infty}\left|\theta_2(x_2,i\sqrt{t})-e^{-x_2\sqrt{t}}\right|\left|\theta_2(x_2,i\sqrt{t})+e^{- x_2\sqrt{t}}\right|dx_2.
\end{align*}
To find estimates for $|\sigma_1|$ and $|\sigma_2|$ we use the  following estimates for large $t$
\begin{equation}\label{theestimates1}
\left|\theta_k(x_k,i\sqrt{t})-e^{-\sqrt{t}x_k}\right|\leq\dfrac{m_1}{\sqrt{t}}e^{-\sqrt{t}x_k}
\end{equation}
and
\begin{equation}\label{theestimates2}
\left|\theta_k(x_k,i\sqrt{t})+e^{-\sqrt{t}x_k})\right|\leq \left(\frac{m_2}{\sqrt{t}}+2\right)e^{-\sqrt{t}x_k}
\end{equation}
This yields the following bounds
\begin{align*}
|\sigma_1|\leq \left|\frac{(a-\sqrt{t} b+1)}{(d-a)\sqrt{t}+bt-c}\right|\frac{m_1(m_2+2\sqrt{t})}{2t^{3/2}}=|-m_1t^{-3/2}+O(t^{-2})|,\\
|\sigma_2|\leq \left|\frac{(1-\sqrt{t} b-d)}{(d-a)\sqrt{t}+bt-c}\right|\frac{m_1(m_2+2\sqrt{t})}{2t^{3/2}}=|-m_1t^{-3/2}+O(t^{-2})|.
\end{align*}
Therefore\begin{align*}
|\sigma_1|=|\sigma_2|=O(t^{-3/2})
\end{align*}
and hence
\begin{align*}
||\mathcal{S}_1||_1=O(t^{-3/2}).
\end{align*}\\

To find estimate on norm of  $\mathcal{S}_2$, we use Lemma \eqref{ssflemma1}. Clearly, $\mathcal{S}_2=(\cdot,f)f-(\cdot,g)g$ is an operator of rank two, where
$$f_k(x_k)=\dfrac{-e^{-x_k\sqrt{t}}}{\sqrt{(d-a)\sqrt{t}+bt-c}}\quad\text{ and }\quad g_k(x_k)=\dfrac{-\theta_k(x_k,i\sqrt{t})}{\sqrt{L(i\sqrt{t})}\theta_k(0,i\sqrt{t})}.$$\\
Let $h=g-f$, then
\begin{align*}
h_k(x_k)&=\dfrac{e^{-x_k\sqrt{t}}}{\sqrt{(d-a)\sqrt{t}+bt-c}}-\dfrac{\theta_k(x_k,i\sqrt{t})}{\sqrt{L(i\sqrt{t})}\theta_k(0,i\sqrt{t})}\\
&=-\left(\dfrac{\theta_k(x_k,\sqrt{t}i)-e^{-\sqrt{t}x_k}}{\sqrt{L(i\sqrt{t})}\theta_k(0,i\sqrt{t})}+\left(\dfrac{1}{\sqrt{L(i\sqrt{t})}\theta_k(0,i\sqrt{t})}-\dfrac{1}{\sqrt{(d-a)\sqrt{t}+bt-c}}\right)e^{-\sqrt{t}x_k}\right).
\end{align*}
To compute estimate on 
$
||h||^2=\sum\limits_{k=1}^{2}\int\limits_0^{\infty}\left|h_k(x_k)\right|^2dx_k
$,
we use \eqref{kasymp} and
\eqref{theestimates1} and after some simplifications we obtain
\begin{align*}
||h||^2\leq \dfrac {m_1^2}{ \left((d-a)\sqrt{t}+bt-c \right) {t}^{3/2}}.
\end{align*}
Similarly,
\begin{align*}
||f||^2=\sum\limits_{k=1}^{2}\int\limits_0^{\infty}\left|\dfrac{-e^{-\sqrt{t}x_k}}{\sqrt{(d-a)\sqrt{t}+bt-c}} \right|^2dx_k=\dfrac {1}{ \left((d-a)\sqrt{t}+bt-c \right)\sqrt{t}}.
\end{align*}
The estimate on the right hand side of inequality \eqref{ssf2} is now given by
\begin{align*}
6||h||^2||f||^2+3||h||^4&\leq\frac {6m_1^2}{ \left((d-a)\sqrt{t}+bt-c \right)^2t^2}+\frac {3m_1^4}{ \left((d-a)\sqrt{t}+bt-c \right)^2t^3}
\\
&=\frac{3m_1^4+6m_1^2t}{\left((d-a)\sqrt{t}+bt-c \right)^2t^3}\\
&=\frac{6m_1^2}{b^2t^4}+\frac{12m_1^2(a-d)}{b^3t^{\frac{9}{2}}}+O(t^{-5})\\
 &=O(t^{-4}).
\end{align*}
This implies
\begin{align*}
||\mathcal{S}_2||_1=O(t^{-2})
\end{align*}
and hence
\begin{align*}
||\mathcal{R}_1||_1=O(t^{-\frac{3}{2}}).
\end{align*}
 \end{proof}

%



\end{document}